\documentclass[12pt,a4paper,final]{article}

\usepackage[utf8]{inputenc}
\usepackage[T1]{fontenc}
\usepackage{amsmath,amssymb,amsfonts,mathtools,amsthm}
\usepackage{geometry}
\usepackage{setspace}
\usepackage{natbib}
\usepackage{booktabs}
\usepackage{multirow}
\usepackage{algorithm}
\usepackage{algpseudocode}
\usepackage{bm}
\usepackage{graphicx}
\usepackage[hidelinks]{hyperref}
\usepackage{subcaption}

\theoremstyle{plain}

\newtheorem{definition}{Definition}

\newtheorem{example}{Example}
\newtheorem{proposition}{Proposition}

\newcommand{\E}{E}
\newcommand{\R}{\mathbb{R}}

\newcommand{\bx}{\bm{x}}

\newcommand{\bX}{\bm{X}}

\newcommand{\bZ}{\bm{Z}}

\title{Distributionally balanced sampling designs}

\author{
    Anton Grafström\thanks{Email: anton.grafstrom@slu.se} 
    \ and 
    Wilmer Prentius\thanks{Email: wilmer.prentius@slu.se} \\[0.3cm]
    \normalsize Department of Forest Resource Management, \\
    \normalsize Swedish University of Agricultural Sciences, Umeå, Sweden
}

\date{\today}

\begin{document}
\maketitle

\setstretch{1}
\begin{abstract}
\noindent
We propose \emph{Distributionally Balanced Designs} (DBD), a new class of probability sampling designs that target representativeness at the level of the full auxiliary distribution rather than selected moments. In disciplines such as ecology, forestry, and environmental sciences, where field data collection is expensive, maximizing the information extracted from a limited sample is critical. More precisely, DBD can be viewed as \emph{minimum discrepancy designs} that minimize the expected discrepancy between the sample and population auxiliary distributions. The key idea is to construct samples whose empirical auxiliary distribution closely matches that of the population. We present a first implementation of DBD based on an optimized circular ordering of the population, combined with random selection of a contiguous block of units. The ordering is chosen to minimize the design-expected energy distance, a discrepancy measure that captures differences between distributions beyond low-order moments. This criterion promotes strong spatial spread, and yields low variance for Horvitz-Thompson estimators of totals of functions that vary smoothly with respect to auxiliaries. Simulation results show that approximate DBD achieves better distributional fit than state-of-the-art methods such as the local pivotal and local cube designs. Hence, DBD can improve the reliability of estimates from costly field data, making distributional balancing effective for constructing representative surveys in resource-constrained applications.
\end{abstract}

\setstretch{1}
\newpage
\section{Introduction}
In modern survey sampling, auxiliary information is increasingly available for the entire population prior to sampling. The fundamental challenge of survey design is to use this information to construct a probability sampling scheme such that estimates vary as little as possible between different samples. Methods have been developed for specific objectives, such as balancing means or ensuring spatial spread, but a unified approach that ensures that the sample is a distributional microcosm of the population has been lacking.

A common approach using auxiliary variables is balanced sampling, exemplified by the cube method \citep{DevilleTille2004}. This approach selects a sample such that the Horvitz-Thompson estimates of the auxiliary totals approximately equal the true population totals. While highly effective when the target variable is linearly related to the auxiliary variables, balanced sampling offers no guaranteed variance reduction for non-linear relationships. If the target depends on the auxiliary variables in a more complex way, balancing solely on means is suboptimal.

In environmental sampling contexts, spatially balanced sampling methods have become increasingly popular. Such methods aim to select samples that are well-spread throughout the auxiliary space. Methods such as the Generalized Random Tessellation Stratified (GRTS) design \citep{StevensOlsen2004}, the Local Pivotal Method (LPM) \citep{GrafstromEtAl2012}, and Balanced Acceptance Sampling (BAS) \citep{Robertson2013} have been shown to be effective in capturing local trends. Other recent approaches to achieve spatial balance include ordering units by an approximate solution to the traveling salesman problem and then applying systematic sampling \citep{DicksonTille2016}, or using weakly-associated vectors to generate spread samples \citep{JauslinTille2020}. Although these methods ensure good coverage, they do not necessarily ensure the best distributional fit.

\citet{GT13} introduced doubly balanced sampling (the local cube method) by combining features of the local pivotal method and the cube method. This method can spread the sample in a set of auxiliary variables and balance the sample on a (possibly different) set of auxiliary variables, and was shown to be close to optimal for a linear model with spatially correlated errors. \citet{Robertson2025} demonstrated that it was efficient to spread and balance on a set of auxiliary variables with empirical tests.  

In this paper, we propose a framework that unifies the concepts of `balanced' and `well-spread' sampling. We posit that the ideal sample is one whose empirical distribution is as close as possible to the population distribution. If the sample distribution closely matches the population distribution, then the sample mean of any sufficiently smooth function of the auxiliary variables will be close to the population mean. This implies simultaneous variance reduction for linear trends, non-linear relationships, and spatial patterns.

We term this framework Distributionally Balanced Designs (DBD). To rigorously quantify distributional discrepancies, we use energy statistics \citep{SzekelyRizzo2013}. Specifically, we want to minimize the energy distance between the possible samples and the population. The energy distance metric belongs to the class of Maximum Mean Discrepancy (MMD) measures and provides a mathematically sound objective that captures differences in all moments, not just the means.

Finding subsets that minimize this distance leads to a combinatorial optimization problem. To make a numerical optimization computationally feasible, we here employ a circular population sequence, a structural approach rooted in the seminal work of \citet{Lahiri1951} and \citet{Lahiri1954} on systematic sampling. By arranging the population in an optimized circular order and selecting a contiguous block with a random start, we reduce the size of the support considerably and maintain equal inclusion probabilities. We use simulated annealing \citep{Kea83} to perform the optimization. \citet{G26} introduced a similar construction for creating master samples with optimized panels and for coordinating the panels in a positive or negative manner.

The main contributions of this paper are:
\begin{enumerate}
    \item We introduce the use of energy distance in probability sampling, providing a rigorous criterion for comparing distributional fit of samples.
    \item We show that the error of the Horvitz-Thompson estimator of functions that vary smoothly with the auxiliary variables is controlled by the distributional discrepancy. 
    \item We provide an optimization algorithm based on simulated annealing with efficient updates that organizes the population into a sequence in which every contiguous block is a representative sample.
    \item We demonstrate via simulations that DBD achieves better distributional fit (lower expected energy distance) than state-of-the-art methods and confirm the reduction of variance.
\end{enumerate}

The remainder of this paper is organized as follows. Section~\ref{sec:method} details the methodology, introduces the optimization objective, defines distributionally balanced designs, and describes the optimization algorithm with efficient updates. Section~\ref{sec:var_est} discusses the estimation of variance. Section~\ref{sec:simulation} presents simulation examples in which DBD is compared with some existing designs. Section~\ref{sec:scalability} discusses computational scalability. Finally, Section~\ref{sec:conclusions} offers concluding remarks.

\section{Methodology} \label{sec:method}

Let $U = \{1, \dots, N\}$ denote the population with auxiliary variables $\bx_i \in \R^p$ for the unit $i\in U$. We consider an equal probability fixed size design with inclusion probabilities $\pi_i=n/N$, where the population is arranged in a circular sequence. Let the vector $\bm{u} = (u_1, u_2, \dots, u_N)$, be a permutation of the indices in $U$. A sample $S$ of size $n$ is selected by drawing a random starting position $j$ uniformly among $\{1, \dots, N\}$ and selecting the contiguous block 
$$s_j = s_j(\bm{u};n)=\{u_{((j+k-1) \text{ mod } N) +1 }|k=1,\dots, n\}$$
in the circular sequence. This design has exactly $N$ possible samples, each with probability $1/N$. Because the start is uniformly random, each unit has the inclusion probability $n/N$. By changing the order of the sequence, we can achieve designs with different properties. See Figure~\ref{fig:design} for a conceptual illustration of the design before and after optimization of the sequence order.
\begin{figure}[htb!]
    \centering
    \includegraphics[width=0.8\linewidth]{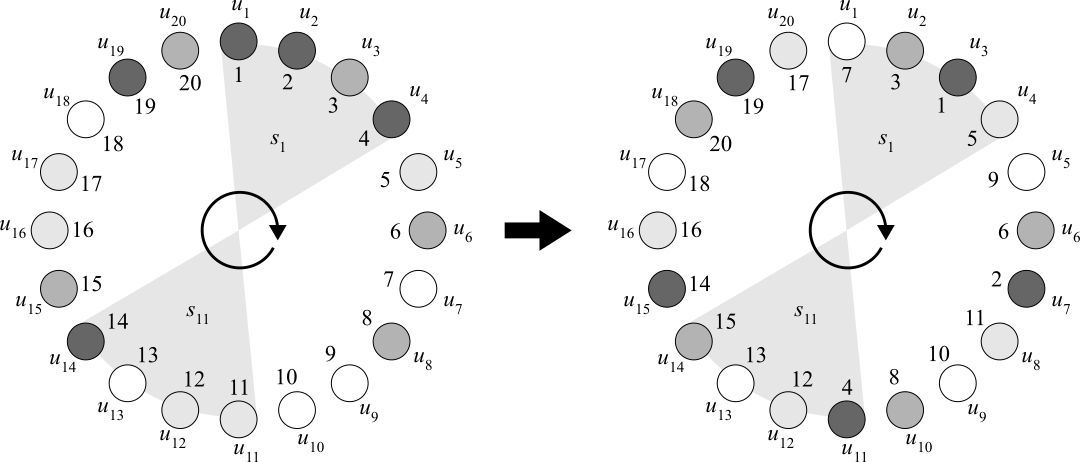}
    \caption{Illustration of the circular sequence design for a population of size $N=20$ and sample size $n=4$, before and after optimization. The nodes represent population units, with grayscale intensity indicating the values of their auxiliary variables. The inner circle of numbers provides the indices of the units in $U$. The samples $s_1$ and $s_{11}$ are shown as shaded sectors. Left: Initial sequence. Many samples selected as a contiguous block are unrepresentative of the population.  Right: The optimized design where the sequence has been reordered. All contiguous blocks of size $n=4$ now provide a good representation of the population distribution.}
    \label{fig:design}
\end{figure}

Let $F_{s_j}$ and $F_U$ denote the empirical distributions of
$\{\bx_i : i \in s_j\}$ and $\{\bx_i : i \in U\}$, respectively.
Explicitly,
\[
F_{s_j} = \frac{1}{n}\sum_{i\in s_j} \delta_{\bx_i},
\qquad
F_U = \frac{1}{N}\sum_{i\in U} \delta_{\bx_i},
\]
where $\delta_{\bx}$ denotes the Dirac measure (unit point mass) at
$\bx \in \mathbb{R}^p$. To measure how well a sample distribution $F_{s_j}$ matches the population distribution $F_U$, we use the energy distance \citep{SzekelyRizzo2013}:
\begin{equation} \label{eq:ed}
    \mathcal{E}(F_{s_j}, F_U) = 2\E\|\bm{X} - \bm{Z}\| - \E\|\bm{X} - \bm{X}'\| - \E\|\bm{Z} - \bm{Z}'\|,
\end{equation}
where $\bm{X}, \bm{X}' \sim F_{s_j}$ are independent draws from the sample distribution, $\bm{Z}, \bm{Z}' \sim F_U$ are independent draws from the population distribution, and $\|\cdot\|$ denotes the euclidean distance. The energy distance is zero if and only if $F_{s_j} = F_U$. 

The term $\E\|\bm{Z} - \bm{Z}'\|$ in \eqref{eq:ed} depends only on the fixed population and is constant across all samples. For the empirical distribution $F_U$, it is given by 
$$\E\|\bm{Z} - \bm{Z}'\|=\frac{1}{N}\sum_{i\in U}\Phi_i,\; \text{where} \; \Phi_i=\frac{1}{N}\sum_{k\in U}\|\bx_i-\bx_k\|.$$
The $\Phi_i$ represents the mean distance from the unit $i$ to the population. 

The other two components of \eqref{eq:ed} create a balance between competing forces. The component $\E\|\bm{X} - \bm{X}'\|$ represents average pairwise distances within the sample. Since $F_{s_j}$ is uniform over the sample:
$$
    \E\|\bm{X} - \bm{X}'\| = \frac{1}{n^2} \sum_{i \in s_j} \sum_{k \in s_j} \|\bx_i - \bx_k\|.
$$
Minimizing \eqref{eq:ed} requires maximizing this term, which forces sample units apart and avoids clustering. The next component is the population fit (attraction), $\E\|\bm{X} - \bm{Z}\|$, which represents the average distances from the sample units to the population mass. Since both $F_{s_j}$ and $F_U$ are uniform:
\begin{equation*}
    \E\|\bm{X} - \bm{Z}\| = \frac{1}{n} \sum_{i \in s_j} \Phi_i,
\end{equation*}
with $\Phi_i$ as previously defined. Minimizing this term while maximizing spread forces the sample to center within the population cloud and to correctly represent its geometric shape.

Note that for the random sample $S$, the empirical distribution
\[
F_S = \frac{1}{n}\sum_{i\in S} \delta_{\bx_i} \quad \text{ satisfies } \quad \E[F_S] = F_U,
\]
where the expectation is taken with respect to the sampling design.
Thus, in the case of equal probabilities, the empirical distribution $F_S$ is
design-unbiased for the population distribution $F_U$. 

We now formalize the meaning of a distributional balanced design that is independent of any specific sampling construction.

\begin{definition}[Distributionally balanced design]
Let $F^\ast$ denote a target auxiliary distribution, and let $\mathcal{P}$ be a class of probability sampling designs with specified inclusion probabilities. 
A design $p^\ast \in \mathcal{P}$ is called \emph{distributionally balanced} if it minimizes the expected distributional discrepancy between the auxiliary-variable distribution of a random sample and the target distribution, i.e., if
\[
p^\ast \in \arg\min_{p \in \mathcal{P}} 
\;E\!\left[ D(F_S, F^\ast) \right],
\]
where $F_S$ is the auxiliary-variable distribution of the sample and $D(\cdot,\cdot)$ is a chosen distributional discrepancy measure. In other words, for a given discrepancy measure, a DBD achieves the best possible approximation to $F^\ast$ within the design class $\mathcal{P}$, rather than exact balance.
\end{definition}

The estimation error of the design-unbiased estimator of the total of functions that vary smoothly with the auxiliary variables is controlled by the distributional mismatch between the sample and population auxiliary distributions. Here, `vary smoothly' is made precise by assuming that the target variable is $y_i = f(\bx_i)$, $i\in U$, for some $f$ in the function class induced by the energy-distance kernel.

\begin{proposition}[Upper bound on the mean square error via energy distance]
\label{prop:rkhs_bound}
Let $k$ be the reproducing kernel associated with the energy distance
\citep[see][]{Sejdinovic2013} so that, for any probability measures $P,Q$
with finite first moment,
\[
\mathcal{E}(P,Q)=\mathrm{MMD}_k^2(P,Q)
=\|\mu_P-\mu_Q\|_{\mathcal{H}_k}^2,
\]
where $\mu_P=\E_{\bX\sim P}[k(\bX,\cdot)]\in\mathcal{H}_k$ denotes the kernel
mean embedding. Let $f\in\mathcal{H}_k$, define $y_i=f(\bx_i)$, and consider
the equal-probability fixed-size estimator
\[
\hat Y=\frac{N}{n}\sum_{i\in S} y_i
\qquad\text{and}\qquad
Y=\sum_{i\in U} y_i.
\]
Then for any sampling design with a random sample $S$ of fixed size $n$,
\begin{equation}
\label{eq:rkhs_mse_bound}
\E\!\left[(\hat Y-Y)^2\right]
\;\le\;
N^2\,\|f\|_{\mathcal{H}_k}^2\;
\E\!\left[\mathcal{E}(F_S,F_U)\right].
\end{equation}
\end{proposition}

See Appendix~\ref{app:proof} for a proof of Proposition~\ref{prop:rkhs_bound}. Thus, minimizing the expected energy distance controls an upper bound on the MSE (and hence the variance) of the Horvitz–-Thompson estimator for target variables that vary smoothly with the auxiliary variables. 

To implement a distributionally balanced design, we must specify the design class $\mathcal{P}$, the discrepancy measure $D$ and the target distribution $F^\ast$. In this paper, we restrict $\mathcal{P}$ to the class of equal-probability designs generated by circular permutations. Furthermore, we employ the energy distance $\mathcal{E}$ as our discrepancy measure $D$, and let $F^\ast=E(F_S)=F_U$. Our objective is therefore to minimize the variability of the realizations $F_{s_j}$ of $F_S$ around $F_U$. We do that by minimizing the expected discrepancy between $F_S$ and $F_U$. The expected discrepancy of a design given by the circular sequence $\bm{u}$ is 
\begin{equation}\label{eq:mean_ed}
 \bar{\mathcal{E}}(\bm{u};n)=\E\left[\mathcal{E}(F_S,F_U)\right] = \frac{1}{N} \sum_{j=1}^N \mathcal{E}(F_{s_j},F_U),  
\end{equation}
where $s_j$ is the sample of size $n$ starting at position $j$ in the sequence $\bm{u}$. Let now $\mathcal{P}$ denote the set of all $N!$ permutations of the indices in $U$. We want to find the sequence that minimizes the expected discrepancy of the design, i.e., to find
\begin{equation} \label{eq:opt}
  \bm{u}^{\ast} \in \underset{\bm{u} \in \mathcal{P}}{\arg\min}
\quad
\bar{\mathcal{E}}(\bm{u};n).
\end{equation}
Since the energy distance is nonnegative and equals zero if and only if the sample and population distributions coincide, minimizing \eqref{eq:mean_ed} directly targets the elimination of distributional variability in the design. Thus, among all circular designs in this class, no alternative construction can achieve a lower expected discrepancy. Due to the factorial growth of the search space, an exhaustive search is computationally intractable for all but the smallest population sizes. Hence, instead of using $\bm{u}^{\ast}$ from \eqref{eq:opt}, we employ simulated annealing to identify a nearly optimal solution $\bm{u}^\circ$.

\subsection{Algorithm}

The Algorithm~\ref{alg:optimization} gives the pseudo-code for selecting the optimized design sequence $\bm{u}^\circ$ from an initial population sequence $\bm{u}$ through simulated annealing. At each iteration, a move is made by swapping the positions of two units in the circular sequence. To ensure connectivity of the permutation space, the first unit is selected uniformly at random, and the second unit is selected uniformly at random from the remaining units. As shown in Appendix~\ref{app:update}, the update of the objective function can be made in $O(n)$ time. An efficient implementation of Algorithm~\ref{alg:optimization} is available in the \texttt{rsamplr} R package \citep{P25}.

\begin{algorithm}[H] 
\caption{Optimization through simulated annealing}
\begin{algorithmic}[1] 
\State \textbf{Input:}
    \State Initial state of the ordered circular population $\bm{u}^{(0)}$, and sample size $n$. 
    \State Iterations $K$, initial temp $T^{(0)}$, and cooling rate $\alpha$.
\State \textbf{Compute:} 
    Expected energy distance $\bar{\mathcal{E}}^{(0)} = \bar{\mathcal{E}}(\bm{u}^{(0)};n)$.
\State \textbf{Initialize:} Best state pair $\bm{u}^\circ \leftarrow \bm{u}^{(0)}$ and $\bar{\mathcal{E}}^\circ \leftarrow \bar{\mathcal{E}}^{(0)}$.
\For{$k = 1$ to $K$}
    \State $\bm{u}^{(k)} \leftarrow \bm{u}^{(k-1)}$.
    \State Select two positions $a,b$ with $a \neq b$, and swap the values of $u^{(k)}_a$ and $u^{(k)}_b$.
    \State Compute expected energy distance $\bar{\mathcal{E}}^{(k)}$.

    \If{$\bar{\mathcal{E}}^{(k)} < \bar{\mathcal{E}}^{\circ}$}
        \State New best state $\bm{u}^\circ \leftarrow \bm{u}^{(k)}$, $\bar{\mathcal{E}}^\circ \leftarrow \bar{\mathcal{E}}^{(k)}$.
    
    \ElsIf{$\bar{\mathcal{E}}^{(k)} \geq \bar{\mathcal{E}}^{(k-1)}$ \textbf{and} $U(0,1) \geq \exp(-(\bar{\mathcal{E}}^{(k)} - \bar{\mathcal{E}}^{(k-1)}) / T^{(k-1)})$}
        \State Reject swap and reset $\bm{u}^{(k)} \leftarrow \bm{u}^{(k-1)}$, $\bar{\mathcal{E}}^{(k)} \leftarrow \bar{\mathcal{E}}^{(k-1)}$.
    \EndIf

    \State Update temperature $T^{(k)} \leftarrow \alpha T^{(k-1)}$.
\EndFor
\State \Return Optimized sequence $\bm{u}^\circ$.
\end{algorithmic} \label{alg:optimization}
\end{algorithm}

Using the optimized circular sequence $\bm{u}^\circ$, a random sample $S$ of size $n$ is selected by drawing a random starting position $j$ uniformly among $\{1, \dots, N\}$ and selecting the contiguous block $s_j$ of length $n$ in the sequence.

\section{Variance estimation} \label{sec:var_est}
Since the proposed design achieves a strong spread in the auxiliary space, the standard variance estimator is not applicable. As the spreading forces many second order inclusion probabilities to be very small (or even zero), the standard variance estimator becomes highly unstable (and may be severely biased in case of second order inclusion probabilities equal to zero).
Instead, we recommend using a local mean variance estimator that approximates the local variance structure. Such estimators are commonly used for spatially balanced samples, and a straightforward version is given by
\begin{equation} \label{eq:varest}
    \hat{V}(\hat{Y}) = N^2 \frac{S_k^2}{n}, \quad \text{where} \quad  S_k^2 = \frac{k}{n(k-1)} \sum_{i \in S} (y_i - \bar{y}_i)^2,
\end{equation}
where $k$ is the size of the neighborhood (including unit $i$), $\bar{y}_i$ is the local mean
$
\bar{y}_i = k^{-1} \sum_{j \in G_i} y_j,
$ and $G_i$ is the set of the $k$ nearest neighbors of unit $i$ in the auxiliary space (including the unit itself).

The local mean variance estimator \eqref{eq:varest} automatically adapts to the structure of the target variable. If nearby units in the auxiliary space tend to have similar 
$y$-values, then each local neighborhood reflects the natural smoothness in the data, and the estimator measures the remaining local variation. If, on the other hand, the target variable does not depend on the auxiliary variables, then the nearest-neighbor groups are essentially random collections of units, and the estimator reflects the overall variation in the population. In this way, the same choice of $k$ works in both situations, and the estimator adjusts itself to the amount of structure present in the data without requiring any tuning.

In practice, values in the range $2 \leq k \leq 4$ typically provide good performance, while setting $k=n$ reproduces the standard variance estimator for independent observations. As $k$ increases toward $n$, the estimator tends to become more conservative. For further discussion, see \citet{GrafstromSchelin2014}.

\section{Simulation examples} \label{sec:simulation}
We illustrate the behavior and performance of the proposed design through three examples. Example~\ref{ex:rate} examines how the expected energy distance decreases as the number of optimization iterations increases. It also assesses robustness by studying the variability of the expected energy distance between repeated optimization runs. Example~\ref{ex:design_comparison} compares the proposed DBD implementation with existing sampling designs in terms of distributional fit, spatial spread, and balance for artificial populations of varying dimensions. Finally, Example~\ref{ex:meuse} compares DBD with existing sampling designs on a real data set, including target variables.

The quality of the designs is assessed using four metrics. For all metrics, lower values indicate better performance. First, we report the mean of the energy distance \eqref{eq:ed} between the sample and population distributions. Second, to measure the spatial spread of the sample, we used the mean of the spatial balance \citep{StevensOlsen2004}. Let $v_i = \sum_{k \in V_i} \pi_k$ be the sum of the inclusion probabilities of the population units within the Voronoi cell $V_i$ of the sample unit $i$. The metric is defined as the mean squared deviation of these sums from their expected value of 1:
\begin{equation*}
    SB(S) = \frac{1}{n} \sum_{i \in S} (v_i - 1)^2.
\end{equation*}
Ideally, a perfectly spread sample partitions the population into regions of equal probability mass ($v_i=1$). A lower expected value of $SB(S)$ indicates a more spatially balanced design.

The local balance (LB) is defined as the root mean squared discrepancy of the balancing equations within the Voronoi partitions \citep{PG24}. Let $\bZ$ be the $N \times (p+1)$ matrix of auxiliary variables augmented with a column of ones, i.e., the row for unit $j$ is $\bm{z}_j = (1, \bx_j^\top)$. For each sample unit $i \in S$, let $V_i$ denote the set of population units closer to $i$ than to any other unit in $S$. The local imbalance vector for the neighborhood $V_i$ is defined as $\bm{e}_i = \pi_i^{-1}\bm{z}_i - \sum_{j \in V_i} \bm{z}_j$, which represents the deviation between the expansion estimator of the neighborhood total and the true neighborhood total. The local balance measure is then
\begin{equation*}
    LB(S) = \sqrt{ \frac{1}{N} \sum_{i \in S} \bm{e}_i^\top \mathbf{Q}^{-1} \bm{e}_i },
\end{equation*}
where $\mathbf{Q} = \bZ^\top \bZ$ is the Gram matrix of the augmented population matrix. A lower expected value of $LB(S)$ indicates that the samples more accurately represent the local density and auxiliary characteristics of their respective neighborhoods. 

Balance deviation (BD) is measured for all auxiliary variables by the euclidean distance between the HT-estimators and the true totals 
\[
BD(S)=\|\hat{\bX}-\bX\|,
\]
where $$\hat{\bX}=\sum_{i \in S}\frac{\bx_i}{\pi_i}  \quad \text{ and } \quad \bX=\sum_{i\in U}\bx_i.$$

\begin{example}[Decay of the expected energy distance and variability across optimizations] \label{ex:rate}
    We generated a synthetic population of size $N=1000$ with $p=5$ auxiliary variables. The auxiliary variables $\bx_i = (x_{i1}, \dots, x_{ip})^\top$ were drawn independently from a uniform distribution on $[0,1]$. The sample size was chosen as $n=50$. Figure~\ref{fig:rate} illustrates the decay of the expected energy distance in an optimization run and compares to the mean of the energy distance achieved from a set of samples selected using the local pivotal method. After about $10^6$ iterations, the curve flattened out. After less than $50{,}000$ iterations, the expected discrepancy of the optimized design is lower than for the pivotal method.
    \begin{figure}[htb!]
    \centering
    \begin{subfigure}[t]{0.66\textwidth}
        \centering
        \includegraphics[width=1\linewidth]{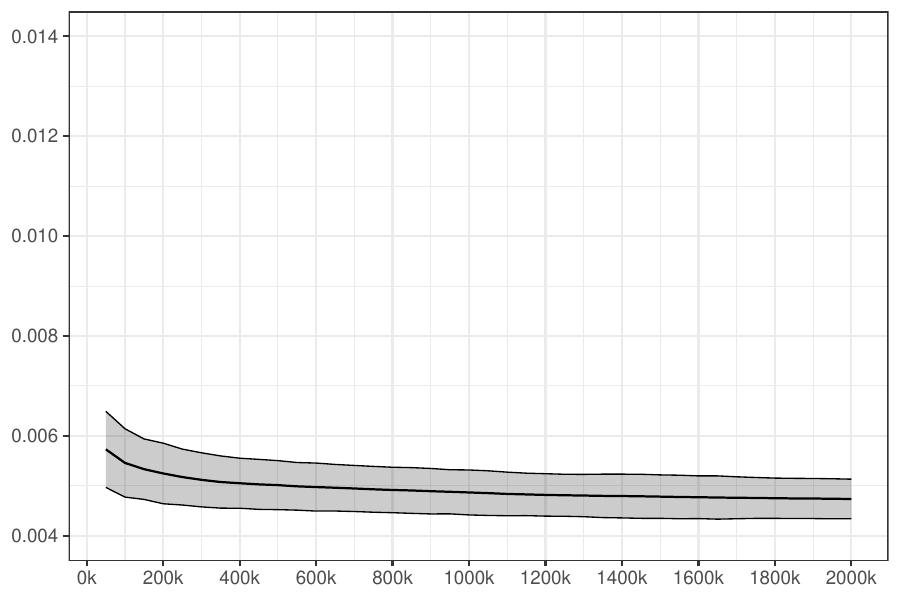}
    \end{subfigure}
    \begin{subfigure}[t]{0.33\textwidth}
        \centering
        \includegraphics[width=1\linewidth]{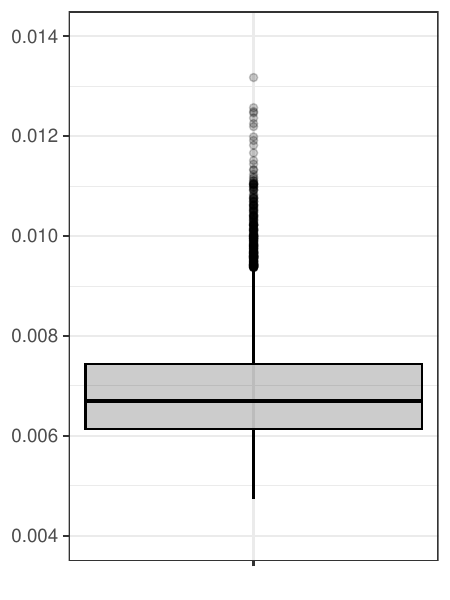}
    \end{subfigure}
    \caption{
        Left panel: expected energy distance of a circular DBD at iterations from 50k-2000k, $\pm 2$ standard deviations. Right panel: box-plot of energy distances of 10000 samples of size $n=50$ selected by the local pivotal method with the same $p=5$ auxiliary variables.
    }
    \label{fig:rate}
\end{figure}

To assess the robustness of the simulated annealing procedure, we performed $50$ independent runs of the optimization algorithm, each starting from a different random initial permutation of the population sequence and using the same annealing schedule, and a total number of iterations at $10^6$. The only source of randomness across runs was the order of the initial sequence and the stochastic decisions within the annealing algorithm.

For each sequence, we computed the expected energy distance $\bar{\mathcal{E}}(\bm{u}^\circ;n)$ as defined in \eqref{eq:mean_ed}. 
The resulting values had a relative standard error of less than $1$\%. 
\end{example}

Example~\ref{ex:rate} indicates that the simulated annealing procedure reliably finds sequences with similar distributional properties. In particular, the variability across optimizations is substantially smaller than the gap between the optimized design and competing designs such as the local pivotal or local cube methods, as seen in the following Example~\ref{ex:design_comparison}. Thus, the proposed distributionally balanced design is not sensitive to the specific realization of the optimization algorithm, and its performance reflects the underlying design criterion rather than a single favorable run.

\begin{example}[Comparisons with some existing designs] \label{ex:design_comparison}
To evaluate the performance of the proposed method across varying degrees of complexity, we generated synthetic populations of size $N=1000$ with $p$ auxiliary variables, where $p \in \{2, 5, 10, 20\}$. The auxiliary variables $\bx_i = (x_{i1}, \dots, x_{ip})^\top$ were drawn independently from a uniform distribution on $[0,1]$. We selected samples of fixed size $n=50$ with equal inclusion probabilities (for $p=5$ we also selected samples of sizes $n=100$ and $n=200$). The proposed method is compared with three designs:
\begin{enumerate}
    \item SRS: Simple random sampling without replacement.
    \item LPM: The local pivotal method \citep{GrafstromEtAl2012}, which selects samples that are well spread in the auxiliary space.
    \item LCube: The local cube method \citep{GT13}, which selects samples that are balanced on the auxiliary totals while maintaining spatial spread.
    \item DBD: The proposed circular distributionally balanced design. The population sequence was optimized using the incremental simulated annealing algorithm with $10^7$ iterations.
\end{enumerate}

For competitor methods (SRS, LPM, LCube), we performed 10,000 independent Monte Carlo repetitions using the \texttt{rsamplr} R package. For DBD, we evaluated the performance on all $N$ possible samples defined by the optimized sequence (i.e., for the complete design). Table~\ref{tab:results_metrics} summarizes the results for dimensions $p=2, 5, 10, 20$ and sample size $n=50$. The resulting distributions of the different metrics under DBD, LCube and LPM are shown in Figure~\ref{fig:metrics_dist}. Table~\ref{tab:results_samplesize} shows the result for $p=5$ auxiliary variables when increasing the sample size to $n=100$ and $n=200$.
\begin{table}[htb!]
\centering
\caption{Simulation results for $N=1000, n=50$ across dimensions $p \in \{2, 5, 10, 20\}$. The table compares the methods on distributional fit using mean energy distance (mean $\mathcal{E}$), spatial balance (mean SB), local balance (mean LB), and balance deviation of auxiliary variables (mean BD). Lower values indicate better performance.
}
\label{tab:results_metrics}
\vspace{0.2cm}
\begin{tabular}{ll rrrr}
\toprule Dims & Method & $\mathcal{E}$ &  SB &  LB & BD \\
\midrule
2  & SRS   &     0.0099  &     0.3375  &     0.1459  &  49.79  \\
   & LPM   &    0.0015  &     0.0879  &     0.0769  &  10.50  \\
   & LCUBE &    0.0013  &     0.0825  &     0.0751  &   7.97  \\
   & DBD   &    0.0010  &     0.0612  &     0.0646  &   4.88  \\
\midrule
5  & SRS   &     0.0167  &     0.2518  &     0.1831  &  84.38  \\
   & LPM   &    0.0069  &     0.1342  &     0.1464  &  36.50  \\
   & LCUBE &    0.0053  &     0.1265  &     0.1429  &  15.07  \\
   & DBD   &    0.0046  &     0.1157  &     0.1391  &  12.44  \\
\midrule
10  & SRS   &     0.0241  &     0.3493  &     0.2739  &  122.96  \\
   & LPM   &    0.0145  &     0.2768  &     0.2566  &  74.54  \\
   & LCUBE &    0.0104  &     0.2702  &     0.2551  &  25.79  \\
   & DBD   &    0.0096  &     0.2629  &     0.2529  &  23.41  \\
\midrule
20  & SRS   &     0.0343  &     0.5651  &     0.4329  &  175.59  \\
   & LPM   &    0.0252  &     0.5151  &     0.4242  &  129.13  \\
   & LCUBE &    0.0171  &     0.5179  &     0.4239  &  45.15  \\
   & DBD   &    0.0167  &     0.5158  &     0.4233  &  41.76  \\
\bottomrule
\end{tabular}
\end{table}
\begin{table}[htb!]
\centering
\caption{Simulation results for $N=1000$ with $p=5$ auxiliary variables at increased sampling fractions ($n=100$ and $n=200$). The table compares the methods on distributional fit using mean energy distance (mean $\mathcal{E}$), spatial balance (mean SB), local balance (mean LB), and balance deviation of auxiliary variables (mean BD). Lower values indicate better performance.
}
\label{tab:results_samplesize}
\vspace{0.2cm}
\begin{tabular}{ll rrrr}
\toprule Size & Method & $\mathcal{E}$ &  SB &  LB & BD \\
\midrule
100  & SRS   &    0.0078  &     0.2560  &     0.1295  &  57.72  \\
     & LPM   &    0.0028  &     0.1400  &     0.1021  &  21.33  \\
     & LCUBE &    0.0021  &     0.1412  &     0.1017  &   7.65  \\
     & DBD   &    0.0019  &     0.1331  &     0.0999  &   6.37  \\
\midrule
200  & SRS   &    0.0035  &     0.2781  &     0.0942  &  38.70  \\
     & LPM   &    0.0011  &     0.1590  &     0.0740  &  12.37  \\
     & LCUBE &    0.0008  &     0.1729  &     0.0761  &   3.89  \\
     & DBD   &    0.0007  &     0.1630  &     0.0742  &   3.12  \\
\bottomrule
\end{tabular}
\end{table}
\begin{figure}[htb!]
    \centering
    \begin{subfigure}[t]{0.95\textwidth}
        \centering
        \includegraphics[width=0.95\linewidth]{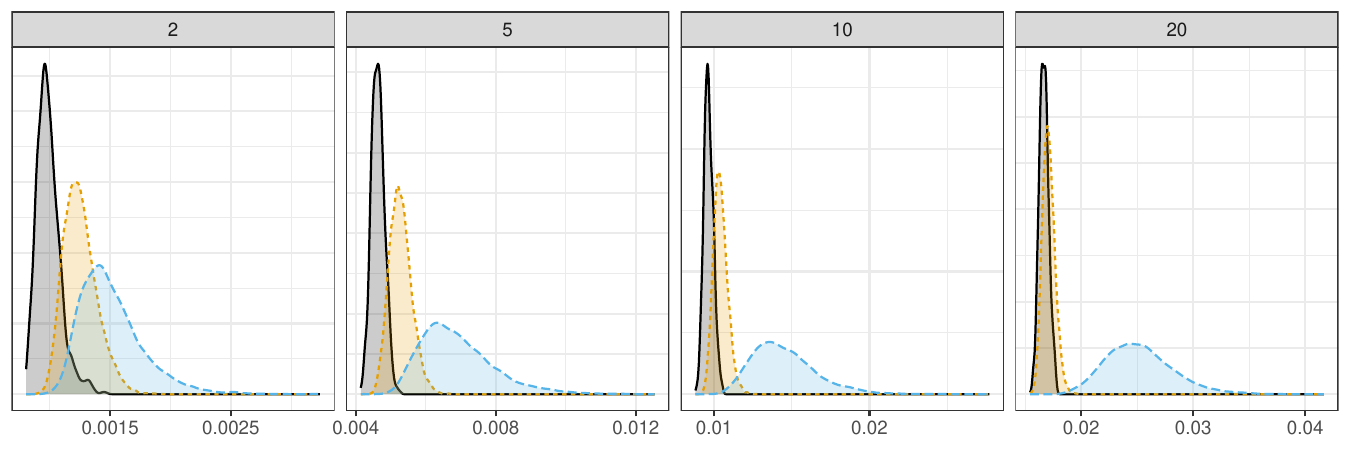}
    \end{subfigure}
    \begin{subfigure}[t]{0.95\textwidth}
        \centering
        \includegraphics[width=0.95\linewidth]{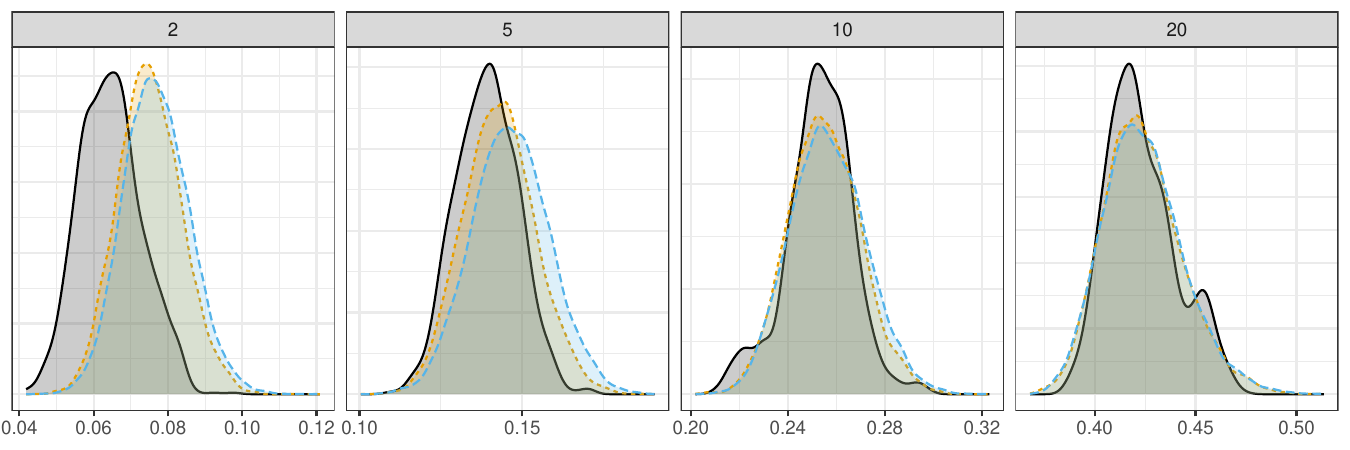}
    \end{subfigure}
    \begin{subfigure}[t]{0.95\textwidth}
        \centering
        \includegraphics[width=0.95\linewidth]{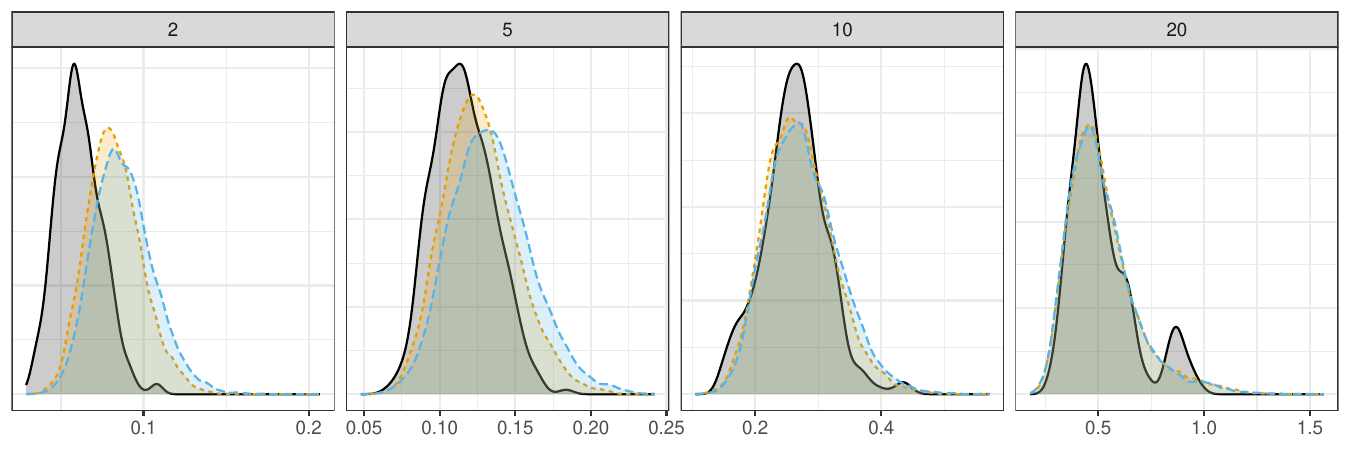}    \end{subfigure}
    \begin{subfigure}[t]{0.95\textwidth}
        \centering
        \includegraphics[width=0.95\linewidth]{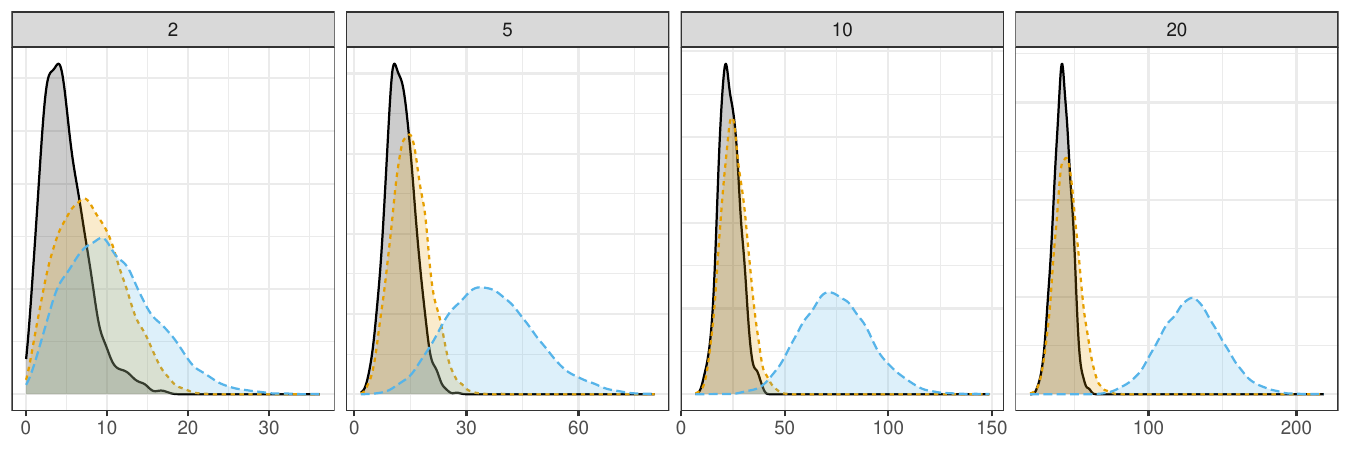}
    \end{subfigure}
    \caption{
    Distributions of the different metrics under three designs with sample size $n=50$. Colors represent the designs: gray is Circular DBD ($10^7$ iterations), orange is LCube, blue is LPM. First row: energy distance. Second row: the local balance measure. Third row: spatial balance. Fourth row: balance deviation. Columns: number of auxiliary variables. 
    }
    \label{fig:metrics_dist}
\end{figure}
\end{example}

From Example~\ref{ex:design_comparison}, we conclude that the optimized distributionally balanced design provides the best distributional fit across all dimensions. We also notice that the local cube method is extremely good at matching the distribution when used to spread and balance on all $p$ variables. The optimization nevertheless yields consistent improvements over the local cube method, and those improvements are  (relatively) greater for lower dimensions. The local pivotal method maintains good spread in all dimensions, but the balancing effect is heavily reduced as the number of dimensions increases. In a spatial context, that is probably not an issue as the target variables usually have spatial trends and do not depend linearly on the coordinates. From Table~\ref{tab:results_samplesize} we note that as the sample size grows, the effect of the improved distributional fit is even greater. For LPM, LCube and Circular DBD, the balance deviation is reducing at a rate close to $1/n$, whereas for SRS the balance deviation reduces at a rate close to $1/\sqrt{n}$. 

\begin{example}[Evaluation using the Meuse dataset] \label{ex:meuse}
We evaluate circular DBD using the Meuse dataset \citep{gstat}, which consists of $164$ locations in a river flood plain. After removal of two locations with missing values, we used the remaining $N=162$ locations as our population. We compare SRS, the local pivotal method (LPM), and the local cube method (LCube) against the proposed DBD for a sample size of $n=20$. All designs, except SRS, use five standardized auxiliary variables: coordinates ($x, y$), elevation (elev), organic matter (org), and copper concentration (Cu). In LCube, all five auxiliary variables are used for both spread and balance. Distances were calculated using standardized auxiliary variables. The target variables are the concentration of zinc (Zn), lead (Pb), and cadmium (Cd).
The DBD sequence was optimized over $10^7$ iterations.

\begin{table}[htb!]
\centering
\caption{Mean energy distance, RRMSE of HT-estimators, and 95\% CI coverage for the Meuse dataset ($N=162, n=20$). For LPM, LCube and circular DBD we applied the local mean variance estimator with $k=2$ neighbors. A total of $10^7$ iterations was used.
}
\label{tab:meuse_results}
\small
\begin{tabular}{ll rrrrrr rrr}
\toprule
&Method & Zn & Pb & Cd & Cu & elev & org & $\mathcal{E}$ & SB & LB \\
\midrule
RRMSE
&SRS  &   0.168 &   0.154 &   0.234 &   0.124 &   0.055 &   0.099 &   0.126 &   0.363 &   0.408\\
&LPM  &   0.102 &   0.094 &   0.120 &   0.053 &   0.016 &   0.043 &   0.044 &   0.270 &   0.176\\
&LCube  &   0.099 &   0.087 &   0.120 &   0.043 &   0.018 &   0.033 &   0.038 &   0.277 &   0.188\\
&DBD  &   0.088 &   0.077 &   0.118 &   0.028 &   0.014 &   0.021 &   0.032 &   0.265 &   0.165\\
\midrule
Coverage
&SRS  &   0.915 &   0.914 &   0.907 &   0.918 &   0.847 &   0.927\\
&LPM  &   0.945 &   0.960 &   0.928 &   0.960 &   0.990 &   0.973\\
&LCube  &   0.961 &   0.972 &   0.942 &   0.990 &   0.943 &   0.995\\
&DBD  &   0.994 &   0.988 &   0.895 &   1.000 &   0.969 &   1.000\\
\bottomrule
\end{tabular}
\end{table}
\end{example}

The results for Example~\ref{ex:meuse}, summarized in Table~\ref{tab:meuse_results}, demonstrate that circular DBD achieves the lowest mean energy distance, indicating superior distributional matching. This discrepancy reduction directly results in efficiency gains. In particular, DBD outperforms the state-of-the-art methods and provides the most accurate estimates for all auxiliary variables and all the target variables (Zn, Pb and Cd). This confirms that by matching the joint distribution of key covariates, DBD preserves the structure of the population. Finally, the coverage results show that DBD tends to maintain safe and conservative statistical inference. Only SRS was consistently below the nominal coverage rate.

\section{Computational scalability} \label{sec:scalability}

The scalability of our implementation of a DBD is determined by two phases: the pre-calculation and the iterative optimization. The pre-calculation involves summing pairwise distances for all units, which scales as $O(N^2)$ in time. This is a one-time cost. Memory requirements remain efficient, as we do not need to store the full distance matrix. 

For the optimization phase, as derived in Appendix~\ref{app:update}, the update strategy ensures that each iteration scales as $O(n)$. However, as $N$ grows, the permutation space expands factorially, typically requiring a super-linear increase in the number of iterations to come close to the global minimum. Empirical tests suggest that for populations up to $N \approx 20,000$, optimization is feasible on a standard desktop CPU in minutes.

For massive populations where direct optimization becomes prohibitive, the proposed method naturally supports a stratified approach. The population $U$ can be partitioned into $H$ disjoint strata (or blocks) $U_1, \dots, U_H$ based on geographical regions or fast clustering methods (e.g., $k$-means), such that the size of each stratum $N_h$ remains manageable (e.g., $N_h \le 10,000$). Since the variance of a stratified estimator is additive, optimizing the sequence $\bm{u}_h$ for each stratum independently preserves most of the variance reduction. This "Block-DBD" approach makes the method linearly scalable for populations of arbitrary size.

\section{Final remarks} \label{sec:conclusions}
We have introduced Distributionally Balanced Designs (DBD), a framework that unifies balanced and well-spread sampling through the principle of distributional matching. By optimizing a circular population sequence to minimize the energy distance between the sample and population distributions, we construct designs that are robust, model-free, and highly efficient. 

The conceptual appeal of DBD lies in the universality of distributional matching as a design principle. Unlike the cube method, which assumes linear relationships, or spatial methods that rely on geometric heuristics, DBD optimizes a direct statistical criterion: the distributional distance between sample and population. Proposition~\ref{prop:rkhs_bound} shows that the HT estimation error for targets that vary smoothly is controlled by the energy distance between the sample and population auxiliary distributions.

The simulation results confirm these theoretical advantages. Circular DBD consistently achieves the best distributional fit across all dimensions tested, with particularly strong performance in lower-dimensional settings. The method simultaneously delivers excellent balance and spatial spread, demonstrating that these properties emerge naturally from distributional matching.

The computational framework makes optimization accessible for populations up to $N \approx 20,000$. For larger populations, the stratified Block-DBD approach achieves linear scalability while preserving most of the variance reduction benefits.

Distributionally balanced designs represent a shift in survey methodology: rather than optimizing isolated properties, we ensure that the samples are close to a distributional microcosm of the population. This approach is particularly useful when multiple auxiliary variables describe ecological or environmental gradients and when the relationship between the target variable and auxiliaries is unknown or nonlinear. Practitioners can immediately start using the implementation we provide in the \texttt{rsamplr} R-package \citep{P25} to select distributionally balanced samples.

Finally, the applicability of DBD extends beyond classical survey sampling. In the era of large-scale machine learning, selecting a representative training subset (coreset) from a massive dataset is a frequent challenge. DBD offers a principled probability-based approach to data reduction that preserves the multivariate distribution of features, potentially improving the generalization of models trained on subsamples.

\appendix

\section{Proof of Proposition \ref{prop:rkhs_bound}} \label{app:proof}
\begin{proof}
Fix any subset $s\subset U$ of size $n$, and let $F_s$ and $F_U$ denote the
empirical distributions of $\{\bx_i:i\in s\}$ and $\{\bx_i:i\in U\}$, respectively.
Write $\bX_s\sim F_s$ and $\bX_U\sim F_U$. Then
\[
\hat Y(s)-Y
=
N\Big(\E[f(\bX_s)]-\E[f(\bX_U)]\Big).
\]
Since $f\in\mathcal{H}_k$, the reproducing property gives
$\E_P[f(\bX)]=\langle f,\mu_P\rangle_{\mathcal{H}_k}$ for any $P$, hence
\[
\Big|\E[f(\bX_s)]-\E[f(\bX_U)]\Big|
=
\Big|\langle f,\mu_{F_s}-\mu_{F_U}\rangle_{\mathcal{H}_k}\Big|
\le
\|f\|_{\mathcal{H}_k}\,\|\mu_{F_s}-\mu_{F_U}\|_{\mathcal{H}_k},
\]
by Cauchy-Schwarz in $\mathcal{H}_k$. Using the energy/MMD equivalence
$\|\mu_{F_s}-\mu_{F_U}\|_{\mathcal{H}_k}^2=\mathcal{E}(F_s,F_U)$ yields
\[
|\hat Y(s)-Y|
\le
N\,\|f\|_{\mathcal{H}_k}\,\sqrt{\mathcal{E}(F_s,F_U)}.
\]
Squaring both sides gives, for every fixed $s$,
\[
(\hat Y(s)-Y)^2
\le
N^2\,\|f\|_{\mathcal{H}_k}^2\,\mathcal{E}(F_s,F_U).
\]
Now let $S$ be the random sample under the design and take expectations:
\[
\E\!\left[(\hat Y-Y)^2\right]
\le
N^2\,\|f\|_{\mathcal{H}_k}^2\;\E\!\left[\mathcal{E}(F_S,F_U)\right],
\]
which proves \eqref{eq:rkhs_mse_bound}.
\end{proof}

\section{Efficient update of the objective function}\label{app:update}

The objective function \eqref{eq:mean_ed} can be written as
\[
\bar{\mathcal{E}}(\bm{u};n)
=
\frac{1}{N}\sum_{j=1}^N \frac{2}{n}\sum_{i\in s_j}\Phi_i-\frac{1}{N}\sum_{j=1}^N\frac{1}{n^2}\sum_{i,k\in s_j}\|\bx_i-\bx_k\| - \frac{1}{N}\sum_{i \in U}\Phi_i.
\]
Since each unit appears in exactly $n$ contiguous windows, the first component (aggregated attraction) simplifies to
\[
\frac{1}{N}\sum_{j=1}^N \frac{2}{n}\sum_{i\in s_j}\Phi_i = \frac{2}{N}\sum_{i\in U}\Phi_i,
\]
which does not depend on the permutation $\bm{u}$. The optimization therefore reduces to maximizing the aggregated repulsion term
\[
R(\bm{u})
=
\sum_{j=1}^N
\sum_{\substack{i<k\\ i,k\in s_j}}
\|\bx_i-\bx_k\|.
\]
Let positions in the circular sequence be indexed by
$r,v \in \{1,\dots,N\}$, and define the circular separation
$t_{rv} =\min\!\big(|r-v|,\; N-|r-v|\big)$. A pair of positions $(r,v)$ is contained together in exactly
\[
w(t_{rv})
=
\begin{cases}
n - t_{rv}, & t_{rv} < n,\\
0, & t_{rv} \ge n.
\end{cases}
\]
contiguous windows of size $n$. Therefore,
\[
R(\bm{u})
=
\sum_{1 \le r < v \le N}
w(t_{rv})\,
\|\bx_{u_r} - \bx_{u_v}\|.
\]
Only pairs with $t_{rv} < n$ contribute to the objective.
Thus, each position interacts with at most $2(n-1)$ of its nearest
neighbors along the circle.

Consider a proposed swap of positions $a$ and $b$ in the sequence.
All pairwise terms in $R(\bm{u})$ that do not involve $a$ or $b$
remain unchanged.
Hence, the change in the objective depends only on pairs
$(a,r)$ and $(b,r)$ with $t_{ar} < n$ or $t_{br} < n$.

Let $\bx_{u_a}$ and $\bx_{u_b}$ denote the auxiliary vectors before the swap,
and let the swapped values be $\bx_{u_a}'=\bx_{u_b}$ and
$\bx_{u_b}'=\bx_{u_a}$.
The change in $R$ is
\[
\Delta R
=
\sum_{\substack{r: \, t_{ar} < n \\ r \ne b}}
w(t_{ar})
\big(
\|\bx_{u_a}' - \bx_{u_r}\|
-
\|\bx_{u_a} - \bx_{u_r}\|
\big)
+
\sum_{\substack{r: \, t_{br} < n \\ r \ne a}}
w(t_{br})
\big(
\|\bx_{u_b}' - \bx_{u_r}\|
-
\|\bx_{u_b} - \bx_{u_r}\|
\big).
\]
Each summation involves at most $2(n-1)$ neighbors and the calculation of $\Delta R$ requires $O(n)$ distance evaluations. Since the full objective function can now be written as
\[
\bar{\mathcal{E}}(\bm{u};n)
=
\frac{1}{N}\sum_{i\in U}\Phi_i
-
\frac{2}{N n^2}\, R(\bm{u}),
\]
the change in the objective of a swap is $\Delta \bar{\mathcal{E}}=-2N^{-1}n^{-2}\Delta R$. Thus, each proposed swap in the simulated annealing procedure can be evaluated in $O(n)$ time.


\begin{thebibliography}{99}

\bibitem[Deville \& Tillé(2004)]{DevilleTille2004}
Deville, J.-C. \& Tillé, Y. (2004). 
Efficient balanced sampling: The cube method. 
\textit{Biometrika}, 91(4), 893-912.

\bibitem[Dickson \& Tillé(2016)]{DicksonTille2016}
Dickson, M.~M. \& Tillé, Y. (2016).
Ordered spatial sampling by means of the traveling salesman problem.
\textit{Computational Statistics}, 31(4), 1359-1372.

\bibitem[Grafström(2026)]{G26}
Grafström, A. (2026). 
Master samples with optimized panels.
\textit{To appear in Survey Methodology}.

\bibitem[Grafström et al.(2012)]{GrafstromEtAl2012}
Grafström, A., Lundström, N. L., \& Schelin, L. (2012). 
Spatially balanced sampling through the pivotal method. 
\textit{Biometrics}, 68(2), 514-520.

\bibitem[Grafström \& Tillé(2013)]{GT13}
Grafström, A., \& Tillé, Y. (2013). Doubly balanced spatial sampling with spreading and restitution of auxiliary totals. \textit{Environmetrics}, 24(2), 120-131.

\bibitem[Grafström \& Schelin(2014)]{GrafstromSchelin2014}
Grafström, A., \& Schelin, L. (2014).
How to select representative samples.
\textit{Scandinavian Journal of Statistics}, 41(2), 277-290.

\bibitem[Jauslin \& Tillé(2020)]{JauslinTille2020}
Jauslin, R. \& Tillé, Y. (2020).
Spatial spread sampling using weakly associated vectors.
\textit{Journal of Agricultural, Biological and Environmental Statistics},
25, 431-451.

\bibitem[Kirkpatrick et al.(1983)]{Kea83} 
Kirkpatrick, S., Gelatt Jr, C. D., \& Vecchi, M. P. (1983). Optimization by simulated annealing. \textit{Science}, 220(4598), 671-680.

\bibitem[Lahiri(1951)]{Lahiri1951}
Lahiri, D. B. (1951).
A method for sample selection providing unbiased ratio estimates.
\textit{Bulletin of the International Statistical Institute}, 33(2), 133-140.

\bibitem[Lahiri(1954)]{Lahiri1954}
Lahiri, D. B. (1954).
On the question of bias of systematic sampling.
\textit{Proceedings of the World Population Conference}, 6, 349-362.

\bibitem[Pebesma(2004)]{gstat}
Pebesma, E.~J. (2004). 
Multivariable geostatistics in S: the gstat package. 
\textit{Computers \& Geosciences}, 30(7), 683-691.

\bibitem[Prentius \& Grafström(2024)]{PG24}
Prentius, W., \& Grafström, A. (2024). How to find the best sampling design: A new measure of spatial balance. \textit{Environmetrics}, 35(7), e2878.

\bibitem[Prentius(2025)]{P25}
Prentius, W. (2025). rsamplr: Sampling Algorithms and Spatially Balanced Sampling. R-package version 0.2.0. https://cran.r-project.org/web/packages/rsamplr/.

\bibitem[Robertson et~al.(2013)]{Robertson2013}
Robertson, B. L., Brown, J. A., McDonald, T., \& Jaksons, P. (2013).
BAS: Balanced acceptance sampling of natural resources.
\textit{Biometrics}, 69(3), 776-784.

\bibitem[Robertson et~al.(2025)]{Robertson2025}
Robertson, B., Price, C., \& Reale, M. (2025). Double dipping with balanced sampling. \textit{Statistics \& Probability Letters}, 110562.

\bibitem[Sejdinovic et~al.(2013)]{Sejdinovic2013}
Sejdinovic, D., Sriperumbudur, B., Gretton, A., \& Fukumizu, K. (2013). 
Equivalence of distance-based and RKHS-based statistics in hypothesis testing. 
\textit{The Annals of Statistics}, 41(5), 2263-2291.

\bibitem[Stevens \& Olsen(2004)]{StevensOlsen2004}
Stevens Jr, D. L., \& Olsen, A. R. (2004). 
Spatially balanced sampling of natural resources. 
\textit{Journal of the American Statistical Association}, 99(465), 262-278.

\bibitem[Székely \& Rizzo(2013)]{SzekelyRizzo2013}
Székely, G. J., \& Rizzo, M. L. (2013). 
Energy statistics: A class of statistics based on distances. 
\textit{Journal of Statistical Planning and Inference}, 143(8), 1249-1272.

\end{thebibliography}
\end{document}